\documentclass[11pt]{amsart}
\usepackage{amsmath,amsfonts, amsthm, amssymb}
\usepackage{microtype}
\usepackage{graphicx}
\usepackage{tikz}
\usepackage{hyperref}

\usepackage[style=alphabetic]{biblatex}
\def\Wigner{{\mathcal W}}
\def\Husimi{{\mathrm{H}}}
\def\Sobolev{{\mathsf{H}}}
\def\balpha{\alpha}

\def\bbC{{\mathbb{C}}}
\def\bbN{{\mathbb{N}}}

\def\Op{\mathrm{Op}}
\def\trace{\operatorname*{trace}}
\def\Id{\mathrm{Id}}
\def\ad{\mathrm{ad}}

\def\Real{{\mathbb R}}

\def\eps{{\varepsilon}}

\newcommand{\diff}{\mathop{}\!\mathrm{d}}

\def\lsim{{\,\lesssim\,}}

\newcommand{\mcal}[1]{{\mathcal{#1}}}

\def\bra#1{\mathinner{\langle{#1}|}}
\def\ket#1{\mathinner{|{#1}\rangle}}
\def\braket#1{\mathinner{\langle{#1}\rangle}}
\def\Braket#1{\!\mathinner{\langle{#1}\rangle}\!}

\numberwithin{equation}{section}

\usepackage{thmtools}
\declaretheorem[numberwithin=section]{theorem}
\declaretheorem[numberlike=theorem]{lemma}
\declaretheorem[numberlike=theorem]{definition}
\declaretheorem[numberlike=theorem]{corollary}
\declaretheorem[numberlike=theorem]{proposition}

\newtheorem{innercustomcor}{Corollary}
\newenvironment{customcor}[1]
  {\renewcommand\theinnercustomcor{#1}\innercustomcor}
  {\endinnercustomcor}

\title{An Egorov Theorem for Wasserstein Distances}

\author{Jordan Cotler}
\address{Department of Physics, Harvard University, Cambridge, MA 02138, USA}
\email{jcotler@fas.harvard.edu}

\author{Felipe Hern\'{a}ndez}
\address{Department of Mathematics, Penn State University, University Park, PA 16802 USA}
\email{felipeh@psu.edu}

\date{\today}
\begin{document}

\begin{abstract}
We prove a new version of Egorov's theorem formulated in the Schr\"odinger picture of quantum mechanics, using the $p$-Wasserstein metric applied to the Husimi functions of quantum states.  The special case $p=1$ corresponds to a ``low-regularity'' Egorov theorem, while larger values $p > 1$ yield progressively stronger estimates.  As a byproduct of our analysis, we prove an optimal transport inequality analogous to the main result of Golse and Paul~\cite{GP2017} in the context of mean-field many-body quantum mechanics.
\end{abstract}

\maketitle

\section{Introduction}
Our universe obeys the laws of quantum mechanics, although our everyday experience portrays a classical world obeying Newtonian laws.  Egorov's theorem forms a mathematical bridge between these two descriptions by demonstrating that quantum observables evolving by the Schr\"{o}dinger equation are well-approximated by observables evolving by
classical equations of motion.
More precisely, Egorov's theorem (see~\cite{hagedornjoye,silvestrov2002ehrenfest,bouzouina2002uniform}, also the textbook reference~\cite[Theorem 11.1]{zworski}) states that for sufficiently smooth observable functions
$a_0\in C^\infty(\Real^{2D})$, one has
\[
e^{i\hat{H}T/\hbar} \Op_\hbar(a_0) e^{-i\hat{H}T/\hbar} = \Op_\hbar(a_T) + O(e^{CT} \hbar)\,,
\]
where $\hat{H}=\Op_\hbar(H)$ is the Weyl quantization of a smooth Hamiltonian $H\in C^\infty(\Real^{2D})$ and $a_t$
solves Hamilton's equation
\[
\partial_t a = \{H,a\}_{\rm PB}\,,
\]
where $\{f,g\}_{\rm PB}$ is the Poisson bracket.

This quantum-classical correspondence naturally aligns with the Heisenberg picture, which focuses on the evolution of \textit{observables} rather than \textit{states} as in the Schr\"{o}dinger picture.  In this paper we prove a new version of Egorov's theorem that uses optimal transport metrics to compare the evolution of quantum and classical states.

\begin{theorem}
\label{thm:Wp-egorov}
Let $H\in C^\infty(\Real^{2D})$ be an \textit{admissible} Hamiltonian as in Definition~\ref{def:admissibleH}.
For any $1 \le p < \infty$ there exists a constant $C_p$, depending only on $p$, $D$ and $H$, such that for all density matrices $\rho$ and all $T \ge 0$,
\begin{equation}
d_{W_p}\big(\Husimi_{\mcal U_T\rho}\,, \Phi_T\Husimi_\rho\big)
\le e^{C_p T}\,\hbar^{1/2},
\end{equation}
where $\Husimi_\rho$ is the Husimi function defined in~\eqref{eq:Husimi-def}, $W_p$ is the
$p$-Wasserstein metric, $\mathcal{U}_T \rho := e^{- i \hat{H} T/\hbar} \rho\, e^{i \hat{H} T/\hbar}$ is the quantum evolution of $\rho$, and $\Phi_T$ is the classical Hamiltonian flow map.
\end{theorem}

This theorem articulates that, after coarse–graining at the natural phase space scale $\sqrt{\hbar}$ (the Husimi smoothing), a quantum state evolves almost classically.  Theorem~\ref{thm:Wp-egorov} is remarkable in that it holds for \textit{arbitrary} quantum states $\rho$, including those lacking good semiclassical descriptions (such as mixtures of coherent states).

To understand the meaning of Theorem~\ref{thm:Wp-egorov}, it is helpful to unpack it in the case $p=1$.
Invoking Kantorovich–Rubinstein duality yields
\[
\left|\int \varphi\,\Husimi_{\mcal U_T\rho} \diff \balpha-\int \varphi\,\Phi_T\Husimi_\rho \diff \balpha\right|
\le \mathrm{Lip}(\varphi)\,e^{C_1 T}\,\hbar^{1/2}
\quad\text{for all $1$-Lipschitz }\varphi,
\]
so expectation values of all Lipschitz classical observables computed from the quantum Husimi density track their classical counterparts with an $O(\sqrt{\hbar})$ error.
Along these lines, letting $\mathcal{S}_{\delta}$ be the standard symbol class (defined in~\eqref{eq:Sdef} below), we have the corollary:
\begin{corollary}
\label{cor:LipEgo}
Let $G\in C^\infty(\Real^{2D})$ be an observable satisfying $\partial_{\alpha_j} G\in \mcal{S}_{\frac12}$ for all $j$.  Then for all $T\in\Real$,
\begin{equation}
\label{eq:LipEgorov}
\left\|\,\mcal{U}_T(\Op_\hbar(G)) - \Op_\hbar(\Phi_T G)\right\|_{L^2\to L^2}
\leq e^{C(D,H) T} \hbar^{1/2}.
\end{equation}
\end{corollary}
The bound~\eqref{eq:LipEgorov} above is a kind of ``low-regularity'' Egorov theorem (since only one derivative of $G$ is quantitatively controlled).
To our knowledge, a bound of this form was not previously known.

Moreover, we comment that the bounds of Theorem~\ref{thm:Wp-egorov} are of increasing strength as $p \geq 1$ is taken large.  That is, simply by H\"{o}lder's inequality and the primal formulation of the $W_p$ distance, one has for any $p\geq 1$ the inequality
\begin{align}
d_{W_1}(\Husimi_{\mcal{U}_T\rho},\Phi_T\Husimi_\rho)
\leq
d_{W_p}(\Husimi_{\mcal{U}_T\rho},\Phi_T\Husimi_\rho).
\end{align}
Thus Theorem~\ref{thm:Wp-egorov} is \textit{strictly stronger} than Corollary~\ref{cor:LipEgo} above, although a natural dual formulation in the Heisenberg picture is not presently clear.

\subsection{Relation to work on mean-field limits}
The use of optimal transport metrics to compare quantum and classical evolutions first appeared in the work of Golse and Paul~\cite{GP2017}, following the pioneering paper of Golse, Mouhot, and Paul~\cite{GMP2016} that introduced quantum optimal transport metrics in many-body quantum physics. (See also~\cite{narnhofer1981vlasov,pezzotti2009mean,graffi2003mean} for earlier works on semiclassical mean field limits.)
In~\cite{GMP2016}, an optimal transport metric $\mcal{E}_\hbar(\rho,\mu)$ is defined to compare quantum states to probability densities, and which is comparable (up to an additive $\hbar^{1/2}$ error) to the metric $d_{W_2}(\Husimi_\rho,\mu)$ that we use.
The objective of~\cite{GP2017} was to compare many-body quantum dynamics to many-body classical dynamics, and an optimal transport approach is necessary to obtain bounds that are uniform as the number of particles $N\to \infty$.
In our setting we take $D=Nd$ where $N$ is the number of particles and $d$ is the spatial dimension of each particle.

Our main result, Theorem~\ref{thm:Wp-egorov}, is unsuitable for the purpose of understanding semiclassical
mean-field limits because it does not provide dimension-independent bounds.
Nevertheless, as a byproduct of our proof we can obtain a result with quantitative bounds that is
comparable to the main result of~\cite{GP2017}.
\begin{theorem}
\label{thm:mean-field}
Let $K,V\in C^{2}(\Real^D)$ and $\hat{H} = K(\hat{p}_1,...,\hat{p}_D) + V(\hat{x}_1,...,\hat{x}_D)$ be a Hamiltonian operator on $L^2(\Real^D)$.  Then
for states $\rho$ of the form
\[
\rho = \int p(\alpha) |\alpha \rangle \langle \alpha |\diff \alpha
\]
where $p$ is a probability density in $L^1(\Real^{2D})$ and $|\alpha\rangle$ is the standard coherent state, we have
\[
\frac{1}{2D} \, d_{W_2}(\Husimi_{\mcal{U}_T\rho}, \Phi_T\Husimi_\rho)^2 \leq (1+2e^{\Lambda T}) \hbar,
\]
where
\[
\Lambda := \frac12 (\|\nabla K\|_{\rm Lip} + \|\nabla V\|_{\rm Lip})\,.
\]
\end{theorem}
For Hamiltonians of the form
\[
\hat{H} = \sum_{j=1}^N \frac{1}{2m_j} \,\hat{\vec{p}}_j^{\,\,2} + \sum_{j=1}^N V_0(\hat{\vec{x}}_j) + \frac1N \sum_{j,k=1}^N V_1(\hat{\vec{x}}_j-\hat{\vec{x}}_k)
\]
arising in mean-field many-particle dynamics with two-particle interactions, we have that $\Lambda$ is uniformly
bounded as $N\to\infty$.

Theorem~\ref{thm:mean-field} is still not directly comparable to the main result of~\cite{GP2017} because it is not
adapted to the bosonic setting, and as stated it does not provide additional information about $n$-particle
marginals for $n\leq N$.  Such a result can be obtained with our methods, but would take us further afield from the proof of our main result Theorem~\ref{thm:Wp-egorov} which we take to be the focus of this paper.

Our approach to proving Theorem~\ref{thm:mean-field} is to track the phase-space localization of the evolution of
coherent states, as measured by re-centered isotropic Sobolev norms.  That is, we consider the quantity
\begin{align}
\label{E:Sobolev1}
\|\psi\|_{\Sobolev^1_\alpha}^2 := \sum_{j=1}^{2D} \|(\hat{\alpha}_j - \alpha_j) \psi\|_{L^2}^2.
\end{align}
This quantity is the quantum analogue of the 2-Wasserstein distance between $\psi$ and a fictitious state perfectly localized at $\alpha$. Indeed, for a probability measure $\mu$ on phase space,
\[
W_2^2(\mu,\delta_\alpha)=\int |z-\alpha|^2\,\mu(\diff z)\,,
\]
so \eqref{E:Sobolev1} provides the operator-theoretic version of the 2-Wasserstein distance to a point mass at $\alpha$.  Since the Sobolev norm~\eqref{E:Sobolev1} is particularly tractable, we can prove Theorem~\ref{thm:mean-field} via a Gronwall inequality applied to the quantity
\[
Q(t) := \|e^{-i\hat{H}t/\hbar} \ket{\alpha_0} \|_{\Sobolev^1_{\alpha(t)}}^2,
\]
where $\alpha(t) =\Phi_t(\alpha_0)$ tracks the classical flow.  This calculation is done in
Section~\ref{sec:localized-evol}.  Fundamentally, this proof does not differ much from the argument of~\cite{GP2017}.
However, because the calculation is simple and leads nicely into the proof of our main result, we include it
for completeness.

\subsection{Proof outline}
\label{sec:outline}
There are two key technical challenges addressed in the proof of Theorem~\ref{thm:Wp-egorov}.

The first, more minor, challenge is that the phase space $\Real^{2D}$ is not compact.
In particular, one does not have a bound of the form $d_{W_p}(\Husimi_{\rho_1}, \Husimi_{\rho_2})
\leq \|\rho_1-\rho_2\|_{\rm tr}$, so it does not suffice to control $\mcal{U}_T \rho = e^{- i H T/\hbar} \rho e^{i H T/\hbar}$ in trace norm.  One needs instead to control the growth
of tails of $\mcal{U}_T \rho$ in phase space.
This necessitates working not with $L^2$ but with isotropic Sobolev spaces
$\Sobolev^k$, defined by
\[
\Sobolev^k := \{\psi\in L^2(\Real^D) \,:\,
\hat{\balpha}^k \psi\in L^2(\Real^D)\}.
\]
Above we use $\hat{\balpha}=(\hat{x}_1,...,\hat{x}_D,\hat{p}_1,...,\hat{p}_D)$
as the phase space operator, and $\hat{\alpha}^k$ is to be interpreted as a
tuple of possible monomials applied to $\psi$.  It is useful to introduce a continuous family of equivalent norms on $\Sobolev^k$ parametrized by the center $\balpha\in\Real^{2D}$,
\[
\|\psi\|_{\Sobolev_\balpha^k} := \|(\hat{\balpha}-\balpha)^k\psi\|_{L^2}.
\]
We show in Section~\ref{sec:localized-evol} that the localized Sobolev norm of a solution to the Schr\"{o}dinger equation with admissible Hamiltonian grows at most
exponentially if $\balpha(t)$ tracks the classical Hamiltonian flow:

\begin{proposition}[Propagation of localization]
\label{prp:localization}
Let $H$ be an admissible Hamiltonian as in Definition~\ref{def:admissibleH},
$\hat{H}=\Op_\hbar(H)$, and $\balpha(T)$ be a trajectory of Hamilton's equations.  Then for any $k$ there exists $C=C(k,D,H)$
such that for all $\psi\in \Sobolev^k(\Real^D)$,
\begin{equation}
\label{eq:general-bd}
\|e^{-i\hat{H}T/\hbar} \psi\|_{\Sobolev^k_{\balpha(T)}}
\leq e^{CT}\|\psi\|_{\Sobolev^k_{\balpha(0)}}.
\end{equation}
Moreover, if $H$ has the form $H = K(p_1,...,p_D) + V(x_1,...,x_D)$ with $K,V\in C^{2}(\Real^D)$, then one more precisely has the bound
\begin{equation}
\label{eq:meanfield-bd}
\Big(\sum_j \|(\hat{\balpha}_j -\balpha_j(T)) e^{-i\hat{H}T/\hbar}\psi\|_{L^2}^2 \Big)
\leq e^{2\Lambda_H T}
\Big(\sum_j \|(\hat{\balpha}_j -\balpha_j(0)) \psi\|_{L^2}^2 \Big).
\end{equation}
where
\[
\Lambda_H := \frac12 (\|\nabla K\|_{\rm Lip} + \|\nabla V\|_{\rm Lip}).
\]
\end{proposition}

Note that above there is no formal connection between the trajectory $\balpha(t)$
and the solution $\psi$, but in practice one obtains the best bound if $\balpha(0)$
minimizes $\|\psi\|_{\Sobolev^k_\balpha}$.
As a corollary of Proposition~\ref{prp:localization}, we obtain a bound for initial data that are mixtures of coherent states.

\begin{corollary}
\label{cor:from-mixture}
Let $p\in L^1(\Real^{2D})$ be a probability density function and
$\hat{\rho} = \int p(\alpha) \ket{\alpha}\bra{\alpha}\diff\alpha$ be a density matrix that is a mixture of coherent states according to $p(\alpha)$.
Then for general admissible Hamiltonians,
\begin{equation}
\label{eq:mixture-general}
d_{W_p}(\Husimi_{\mcal{U}_T\rho}, \Phi_T \Husimi_{\rho}) \leq e^{C_pT} \hbar^{1/2}.
\end{equation}
If instead $\hat{H}=K(\hat{p}_1,...,\hat{p}_D) + V(\hat{x}_1,...,\hat{x}_D)$ we have the bound
\begin{equation}
\label{eq:mixture-meanfield}
\frac{1}{2D}\, d_{W_2}(\Husimi_{\mcal{U}_T\rho}, \Phi_T \Husimi_{\rho})^2 \leq (1 + 2 e^{\Lambda_HT}) \hbar.
\end{equation}
\end{corollary}
The derivation of Corollary~\ref{cor:from-mixture} from Proposition~\ref{prp:localization} is done
in Section~\ref{sec:pffrommixture}.  Note that Corollary~\ref{cor:from-mixture} already implies Theorem~\ref{thm:mean-field}.

The second, and more serious, challenge addressed in the proof of Theorem~\ref{thm:Wp-egorov} is that one must handle
\textit{arbitrary} states $\rho$.  Indeed, in general density matrices $\rho$ can encode long-range superpositions that manifest in
high-frequency oscillations in the Wigner function $\Wigner_\rho$.  These oscillations are not robustly encoded in the Husimi function $\Husimi_\rho$, so part of the content of Theorem~\ref{thm:Wp-egorov} is the conclusion that long-range superpositions in the initial data do not immediately affect the dynamics of the coarse-grained phase
space distribution.

To overcome this more serious difficulty we first make a small reduction.  We use the noising channel
(referring the reader to Section~\ref{sec:husimi} for the definition of the coherent state $\ket{\alpha}$)
\begin{align}
\label{E:noisesuper1}
\mcal{N}[\rho] := (2\pi\hbar)^{-D} \int \Husimi_\rho(\alpha) \ket{\alpha}\bra{\alpha}\diff \alpha,
\end{align}
which can alternatively be described as a mixture of phase-space translations of $\rho$,
\begin{equation}
\label{eq:translation-mixture}
\mcal{N}[\rho] = \int \mcal{T}_\alpha[\rho]  \,\gamma_\hbar(\alpha)\diff \alpha.
\end{equation}
Above, $\mcal{T}_\alpha$ is the phase-space translation operator defined in~\eqref{eq:Talpha-def} and $\gamma_\hbar$ is a Gaussian of variance $\hbar$.
By definition, $\mcal{N}[\rho]$ has the form required in Corollary~\ref{cor:from-mixture}.  That is, we have
\[
d_{W_p}(\Phi_T \Husimi_{\mcal{N}[\rho]}, \Husimi_{\mcal{U}_T \circ \mcal{N}[\rho]}) \leq e^{CT} \hbar^{1/2}.
\]

By the triangle inequality, what remains is to bound
\[
d_{W_p}(\Phi_T \Husimi_{\mcal{N}[\rho]}, \Phi_T \Husimi_\rho)
+ d_{W_p} (\Husimi_{\mcal{U}_T[\rho]}, \Husimi_{\mcal{U}_T\circ\mcal{N}[\rho]}).
\]
The first term is bounded by $e^{CT}\hbar^{1/2}$ because $\Phi_T$ is $e^{CT}$-Lipschitz.
The second term, on the other hand, is a genuinely quantum-mechanical quantity.  By convexity of the $W_p$ transport
distance and the identity~\eqref{eq:translation-mixture}, we have
\[
d_{W_p}( \Husimi_{\mcal{U}_T[\rho]}, \Husimi_{\mcal{U}_T\circ \mcal{N}[\rho]})^p
\leq \int \gamma_\hbar(\alpha)
d_{W_p}( \Husimi_{\mcal{U}_T[\rho]}, \Husimi_{\mcal{U}_T\circ \mcal{T}_\alpha[\rho]})^p \diff \alpha\,.
\]
Trivially we have $\mcal{U}_T\circ \mcal{T}_\alpha[\rho] = \mcal{U}_T\circ \mcal{T}_\alpha \circ \mcal{U}_{-T}
[\,\mcal{U}_T\rho]$.  That is, we need the following bound for general density matrices $\nu$:
\[
d_{W_p}(\Husimi_{\nu}, \Husimi_{\mcal{U}_T\circ\mcal{T}_\alpha\circ\mcal{U}_{-T}[\nu]}) \leq e^{CT}|\alpha|\,.
\]
Note that the operator $\mcal{U}_T\circ\mcal{T}_\alpha\circ\mcal{U}_{-T}[\nu]$ has the form
\[
\mcal{U}_T\circ\mcal{T}_\alpha\circ\mcal{U}_{-T}[\nu] = W_{T,\alpha} \,\nu\, W_{T,\alpha}^\dagger\,,
\]
with
\[
W_{T,\alpha} := e^{i\hat{H}T/\hbar} \tau_\alpha e^{-i\hat{H}T/\hbar}.
\]
As a consequence of Proposition~\ref{prp:localization}, this operator is well-localized in phase space in the sense that $\braket{\beta | W_{T,\alpha}|\beta'}$ is rapidly decaying in $|\beta-\beta'|$.

Then the main result will follow from the following general result:
\begin{proposition}[Transport by local unitaries, informally]
\label{prp:local-transport-informal}
Let $W$ be a unitary operator that is \textit{well-localized in phase space}.  Then for
arbitrary states $\rho$,
\[
d_{W_p}(\Husimi_\rho, \Husimi_{W\rho W^\dagger}) \leq C(p,W) \hbar^{1/2}.
\]
\end{proposition}
Proposition~\ref{prp:local-transport-informal} is stated more precisely and proven in Section~\ref{sec:W-transport}.
The proof uses the Kantorovich duality
for $W_p$.  After some simple manipulations, one can reduce Proposition~\ref{prp:local-transport-informal} to the statement
that the operator
\[
\Op^w(f\ast \gamma_\hbar) + W\Op^w(g\ast \gamma_\hbar)W^\dagger + C(p,W)\hbar^{1/2}
\]
(where $\Op^w$ is defined in~\eqref{E:Opwdef1}) is positive-semidefinite for $f$ and $g$ satisfying
\[
f(\alpha) + g(\beta) \leq |\alpha-\beta|^p.
\]
To prove this we develop a functional-analytic framework for proving positive semi-definiteness based on Lemma~\ref{lem:abstract-positivity}, as  explained in Section~\ref{sec:mainthm}.

\subsection{Organization of the paper}
In Section~\ref{sec:notation} we recall notation and preliminary results from semiclassical analysis.
Our results about growth of isotropic Sobolev norms are proved in Section~\ref{sec:localized-evol}, including the proof of Proposition~\ref{prp:localization} and Corollary~\ref{cor:from-mixture}.  Then in Section~\ref{sec:W-transport} we prove Proposition~\ref{prp:local-transport-informal}.  Finally in Section~\ref{sec:proof}
we complete the proofs of Theorem~\ref{thm:Wp-egorov} and Corollary~\ref{cor:LipEgo}.

\section{Notation and Preliminaries}
\label{sec:notation}

This Section introduces our notation for quantum mechanics and summarizes standard semiclassical analysis results.

\subsection{Phase space notation}
Points in phase space $\alpha\in\Real^{2D}$ are represented by Greek letters.  We write $\alpha=(\alpha_x,\alpha_p)$ where
$\alpha_x,\alpha_p\in\Real^D$.  The standard symplectic form on $\Real^{2D}$ is
\[
\omega = \begin{pmatrix}0 & \Id_{D} \\ -\Id_D & 0\end{pmatrix}.
\]
Phase space variables are indexed by $a\in[2D] = \{1,...,2D\}$, with the indices $1\leq a\leq D$ denoting position variables and $D+1\leq a\leq 2D$ momentum variables.  Accordingly, we write $\partial_a$ for the corresponding phase space partial derivative.

The symplectic form is used to define the Poisson bracket of two functions $f,g\in C^1(\Real^{2D})$:
\[
\{f,g\}_{\text{PB}} := \sum_{j,k} \partial_j f \,\omega_{jk}\, \partial_k g = (\nabla f)^\intercal \cdot \omega \cdot \nabla g\,.
\]

\subsection{Phase space operators and moments}

We denote the position operators on $\mathbb{R}^D$ by $\hat{x} = (\hat{x}_1,\ldots,\hat{x}_D)$, where $\hat{x}_j$ acts as multiplication by $x_j$. The momentum operators are $\hat{p} = (\hat{p}_1,\ldots,\hat{p}_D)$ with $\hat{p}_j = -i\hbar \partial_j$. For convenience, we introduce the phase space operator $\hat{\alpha} = (\hat{x},\hat{p})$ with components $\hat{\alpha}_a$ indexed by $a\in[2D]$:
\[
\hat{\alpha}_a :=
\begin{cases}
\hat{x}_a, &1\leq a\leq D \\
\hat{p}_{a-D}, &D+1\leq a\leq 2D
\end{cases}.
\]

We also need to work with moments of the phase space operators.  A \textit{monomial} of phase space operators is specified by
a tuple $\vec{m} = (m_1,...,m_k)\in[2D]^k$, and we write
\[
\hat{\alpha}_{\vec{m}} = \hat{\alpha}_{m_1} \hat{\alpha}_{m_2} \cdots \hat{\alpha}_{m_k}.
\]
Given a point in phase space $\alpha\in\Real^{2D}$ we also use the centered monomials $(\hat{\alpha}-\alpha)_{\vec{m}}$
which simply replace each $\hat{\alpha}_a$ above by $\hat{\alpha}_a-\alpha_a$.

The \textit{order} of a monomial, written $n(\vec{m})\in \bbN^{2D}$, is the tuple recording the number of appearances
of each index in $\vec{m}$.  For example in $D=2$ we have $n((1,3,4,3)) = (1,0,2,1)$, since in $(1,3,4,3)$ the number `$1$' appears one time, `$2$' appears zero times, `$3$' appears two times, and `$4$' appears one time.  We write $|n|(\vec{m})$ for the
total degree $|n|(\vec{m}) := \sum_a n_a(\vec{m})$.

The phase space operators satisfy the commutation relations
\[
[\hat{\alpha}_a,\hat{\alpha}_b] = i \hbar \omega_{ab}\,.
\]
Commutators against phase space operators are common enough that this also has a special notation.  In particular, given $a\in[2D]$
we write
\[
\ad_a[\hat{A}] := [\hat{\alpha}_a, \hat{A}]\,.
\]
Given a multi-index $I\in\bbN^{2D}$ we write the power
\[
\ad^I = \prod_{a=1}^{2D} \ad_a^{I_a}\,.
\]
Note that there is no operator ordering issue above because $\ad_a \ad_b = \ad_b\ad_a$ for our phase space operators.

\subsection{Quantization and Symbol classes}
Given a smooth function $A\in C^\infty(\Real^{2D})$ and $\hbar > 0$ we define the Weyl quantization $\Op_\hbar(A)$ to be the operator given by the formula
\[
\Op_\hbar(A)f(x) = (2\pi \hbar)^{-D}\int e^{i  p \cdot (x-y)/\hbar } A(\frac{x+y}{2},p) f(y) \diff y \diff p.
\]
Note that with this definition, $\Op_\hbar(x^k) = \hat{x}^k$ and $\Op_\hbar(p^k) =\hat{p}^k$.  We sometimes abuse notation and
simply write $\hat{A} = \Op_\hbar(A)$.  We make use of the identity
\[
\ad_a[ \Op_\hbar(A)] = i\hbar \sum_b \Op_\hbar( \omega_{ab}\, \partial_b A).
\]

We will often quantize functions with certain smoothness and decay conditions, which we encapsulate by symbol classes.  The standard
symbol class $\mcal{S}_\delta$ is given by
\begin{equation}
    \label{eq:Sdef}
    \mcal{S}_\delta := \{f\in C^\infty(\Real^{2D}) \,:\, \|\partial^j f\|_{L^\infty} \leq C_j \hbar^{-\delta j}\}.
\end{equation}
For $0\leq \delta\leq \frac12$, symbols in $\mcal{S}_\delta$ correspond to uniformly bounded operators by the Calder\'on-Vaillancourt theorem.
In particular, we use the following fact about such symbols:
\begin{proposition}
\label{prp:comm-ests}
Let $0\leq \delta\leq \frac12$ and $A\in\mcal{S}_\delta$.  Then $\hat{A}$ is bounded on $L^2$ and for any multi-index $I\in\bbN^{2D}$,
\begin{equation}
\label{eq:comm-symbol}
\|\ad^I [\hat{A}]\|_{\rm op} \leq C_{|I|} \hbar^{(1-\delta)|I|}.
\end{equation}
\end{proposition}
Beals' lemma~\cite[Theorem 8.3]{zworski} shows that bounds of the form~\eqref{eq:comm-symbol} \textit{characterize} the symbol class $\mcal{S}_\delta$,
but we do not use this fact.

We define admissible Hamiltonians in terms of the symbol class $\mcal{S}_\delta$.
\begin{definition}[Admissible Hamiltonians]
\label{def:admissibleH}
A Hamiltonian $H\in C^\infty(\Real^{2D})$ is \textit{admissible}
if $\partial_{ab} H \in \mcal{S}_{\frac12}$ for every second order partial derivative.  That is, $H$ is
admissible if for every multi-index $J$ of order $|J|\geq2$,
\[
\|\partial_\alpha^J H\|_{L^\infty} \leq C_J \hbar^{1-|J|/2}.
\]
\end{definition}

\subsection{Wigner and Husimi functions}
\label{sec:husimi}
 To each operator $A:L^2(\Real^D) \to L^2(\Real^D)$ with Schwartz kernel $K_A$ we can associate a tempered distribution $\Wigner_A$ (which is a bounded function when $A$ is trace-class) by
\[
\Wigner_A(x,p) = (2\pi\hbar)^{-D} \int e^{i p \cdot(x-y)/\hbar}  K_A(x+y/2,x-y/2)  \diff y.
\]
Given an operator $A$ and a function $f$ we have the identities
\[
\Op_\hbar(\Wigner_A) = (2\pi \hbar)^{-D} A.
\]
We therefore write $\Op_\hbar^{-1} = (2\pi \hbar)^D \Wigner_A$.  Similarly
\[
\Wigner_{\Op_\hbar(f)} = (2\pi\hbar)^{-D} f,
\]
so we write $\Wigner^{-1} := (2\pi\hbar)^D \Op_\hbar$.  With these definitions we have the identity
\begin{equation}
\label{eq:wigner-weyl}
\trace[\Op_\hbar(f) A] = \langle f, \Wigner_A\rangle_{L^2(\Real^{2D})},
\end{equation}
which along with the identity $\Op_\hbar(1) = \Id$ implies that $\trace[\rho] = \int \Wigner_\rho(x,p)\diff x\diff p$.

For each point in phase space, $\alpha\in\Real^{2D}$, we write
$\ket{\alpha} \in L^2(\Real^D)$ to denote the isotropic coherent state centered at $\alpha$.  For concreteness, $\ket{\alpha}$ is
chosen such that
\[
\langle x | \alpha \rangle = (2\pi \hbar)^{-D/4} \exp(i \alpha_p\cdot(x-\alpha_x/2)/\hbar)\exp(-|x-\alpha_x|^2 / (2\hbar)).
\]
With this choice, one has $\ket{\alpha} = \tau_\alpha \ket{0}$ where $\tau_\alpha$ is the phase space translation operator
$\tau_\alpha := \exp(i (\alpha_p \cdot\hat{x}-  \alpha_x \cdot \hat{p})/\hbar)$ and $\ket{0}$ is the ground state of the quantum harmonic oscillator with unit mass in $D$ dimensions.  We then define the
translation $\mcal{T}_\alpha$ operator for \textit{states} by
\begin{equation}
\label{eq:Talpha-def}
\mcal{T}_\alpha[\rho] = \tau_\alpha \,\rho\, \tau_{-\alpha}.
\end{equation}
Moreover, the states $\ket{\alpha}$ decompose the identity as follows:
\[
\Id = (2\pi \hbar)^{-D} \int_{\Real^{2D}} |\alpha\rangle \langle \alpha |\,\diff\alpha\,.
\]

The Wigner function of the pure state $\ket{\alpha}\bra{\alpha}$ is also a Gaussian, given by
\begin{equation}
\label{eq:gaussian-wigner}
\Wigner_{\ket{\alpha}\bra{\alpha}}(x,p) = (2\pi\hbar)^{-D} \exp(-|(x,p)-\alpha|^2/(2\hbar)).
\end{equation}
Given any quantum state $\rho$ we define the Husimi function
\begin{equation}
\label{eq:Husimi-def}
\Husimi_\rho(\alpha) := (2\pi\hbar)^{-D}\braket{\alpha| \rho |\alpha}\,,
\end{equation}
which is a probability distribution over $\mathbb{R}^{2D}$.  Combining~\eqref{eq:wigner-weyl} with~\eqref{eq:gaussian-wigner} we obtain the identity
\[
\Husimi_\rho = \Wigner_\rho \ast \gamma_{\hbar\,\Id},
\]
where $\gamma_\sigma$ is the $L^1$-normalized Gaussian on $\Real^{2D}$ with covariance $\sigma$.

We also will make use of the noising superoperator $\mcal{N}[\rho] := (2\pi\hbar)^{-D} \int \Husimi_\rho(\alpha) \ket{\alpha}\bra{\alpha}\diff \alpha$ that we defined in~\eqref{E:noisesuper1}. We note that the noising operator satisfies the identity
\begin{equation}
\label{eq:noising-identity}
\Wigner_{\mcal{N}[\rho]} = \Wigner_\rho \ast \gamma_{2\hbar\,\Id}\,.
\end{equation}

\section{Propagation of Local States}
\label{sec:localized-evol}

In this Section we prove Proposition~\ref{prp:localization}.  There are two claims in this proposition,
one pertaining to general Hamiltonians $H$ with $\partial_{ab}H\in\mcal{S}_{\frac12}$, and the other to
Hamiltonians of the form $H=K(p)+V(x)$.  First in Section~\ref{sec:KV} below we handle the latter case
and prove as a consequence~\eqref{eq:meanfield-bd}.  Then in Section~\ref{sec:Hgenloc} we handle the former
case and prove~\eqref{eq:general-bd}.

\subsection{Hamiltonians of the form $K(p)+V(x)$}
\label{sec:KV}
First we prove Proposition~\ref{prp:localization} in the case that $H=K(p)+ V(x)$.  In this case we can prove
a quantitative bound in the isotropic Sobolev space $\Sobolev^1$.  Let $\psi\in L^2(\Real^D)$ and
$\alpha(t)$ be a trajectory of Hamilton's equations in $\Real^{2D}$.  Then we define the quantity
\[
Q(t) :=
\|e^{-i\hat{H}t/\hbar}\psi\|_{\Sobolev^1_{\alpha(t)}}^2  := \| (\hat{\alpha}-\alpha(t)) e^{-i\hat{H}t/\hbar}\psi\|_{L^2}^2
= \sum_{a=1}^{2D} \|(\hat{\alpha}_{a} - \alpha_a(t)) e^{-i\hat{H}t/\hbar}\psi\|_{L^2}^2.
\]
The bound~\eqref{eq:meanfield-bd} will follow from the estimate
\[
\frac{d}{dt} Q(t) \leq (\|\nabla K\|_{\rm Lip} + \|\nabla V\|_{\rm Lip}) Q(t).
\]
Without loss of generality it suffices to prove this bound at $t=0$ (as otherwise we just conjugate by the
time evolution, redefining $\psi$ to be the time evolution up to time $t$).  To compute the time derivative first
we observe, writing $\hat{A}(t) := e^{i\hat{H}t/\hbar} \hat{A} e^{-i\hat{H}t/\hbar}$, that
\begin{align*}
\left. \frac{d}{dt}\hat{x}_j(t)\right|_{t=0} &= \partial_j K(\hat{p}) \\
\left. \frac{d}{dt}\hat{p}_j(t)\right|_{t=0} &= -\partial_j V(\hat{x}).
\end{align*}
The idea is to use a first order Taylor approximation for each of these equations, as described in the following lemma.
\begin{lemma}
\label{lem:taylor}
If $F\in C^1(\Real^D)$ then
\begin{equation}
\label{eq:F-remainder}
F(x) = F(0) + \sum_j x_j R_j(x)
\end{equation}
with
\begin{equation}
\label{eq:Rjdef}
R_j(x) = \int_0^1 \partial_j F(sx)\diff s\,.
\end{equation}
\end{lemma}
\begin{proof}
This is simply the fundamental theorem of calculus.
\end{proof}
Therefore, for any $x_0,p_0\in\Real^D$ there are functions $R_{jk;p_0}^K$ and $R_{jk;x_0}^V$ such that
\begin{align*}
\left. \frac{d}{dt}\hat{p}_j(t)\right|_{t=0} &= \partial_j K(p_0) + \sum_k (\hat{p}_k-(p_0)_k) R_{jk;p_0}^K(\hat{p}) \\
\left. \frac{d}{dt}\hat{p}_j(t)\right|_{t=0} &= -\partial_j V(x_0) + \sum_k (\hat{x}_k-(x_0)_k) R_{jk;x_0}^V(\hat{x})\,.
\end{align*}
Then writing $\alpha(t)= (x(t),p(t))$, using the equations of motion for $x(t)$ and $p(t)$, and
applying~\eqref{eq:F-remainder} with $F=\partial_j V$, we have
\begin{align*}
\left.\frac{d}{dt} (\hat{x}_j(t)-x_j(t))\right|_{t=0} &= \sum_{k=1}^D (\hat{p}_k - p_k) R^K_{jk;p}(\hat{p}) \\
\left.\frac{d}{dt} (\hat{p}_j(t)-p_j(t))\right|_{t=0} &= -\sum_{k=1}^D (\hat{x}_k - x_k) R^V_{jk;x}(\hat{x}).
\end{align*}
We can make this more compact by using matrix notation, writing $R^K$ and $R^V$ for matrix-valued functions on $\Real^D$:
\begin{align*}
\left.\frac{d}{dt} (\hat{x}(t)-x(t))\right|_{t=0} &= R^K(\hat{p}) \cdot (\hat{p} - p)  \\
\left.\frac{d}{dt} (\hat{p}(t)-p(t))\right|_{t=0} &= - R^V(\hat{x}) \cdot (\hat{x} - x).
\end{align*}
The operators $R^K(\hat{p})$ and $R^V(\hat{x})$ acting on $L^2(\Real^D)\otimes \bbC^D$ satisfy the operator norm bounds
\begin{align*}
\|R^K(\hat{p})\|_{L^2(\Real^D) \otimes \bbC^D \to L^2(\Real^D)\otimes \bbC^D}
&\!=\! \sup_{p\in\Real^D} \|R^K(p)\|_{\bbC^D\to\bbC^D} \!\leq\! \sup_{p\in\Real^D} \|\operatorname*{Hess} K\|_{\bbC^D\to\bbC^D}
= \|\nabla K \|_{\rm Lip}\\
\|R^V(\hat{x})\|_{L^2(\Real^D) \otimes \bbC^D \to L^2(\Real^D)\otimes \bbC^D}
&\!=\! \sup_{x\in\Real^D} \|R^V(x)\|_{\bbC^D\to\bbC^D} \!\leq\! \sup_{x\in\Real^D} \|\operatorname*{Hess} V\|_{\bbC^D\to\bbC^D}
= \|\nabla V \|_{\rm Lip} \,.
\end{align*}
Here we bounded $R^K$ by $\operatorname*{Hess} K$ using the triangle inequality over the integral~\eqref{eq:Rjdef}, and similarly for $R^V$. Observing that $\|(\hat{\alpha}-\alpha(t))e^{-i\hat{H}t/\hbar}\psi\|_{L^2} = \|(\hat{\alpha}(t)-\alpha(t))\psi\|_{L^2}$,
we can now compute
\begin{align*}
\left.\frac{d}{dt} Q(t) \right|_{t=0}
&= 2 \langle (\hat{x}-x)\psi, R^K(\hat{p}) (\hat{p}-p) \psi\rangle_{L^2(\Real^D)\otimes \bbC^D}
+ 2 \langle (\hat{p}-p)\psi, R^V(\hat{x}) (\hat{x}-x) \psi\rangle_{L^2(\Real^D)\otimes \bbC^D} \\
&\leq 2 (\|\nabla K\|_{\rm Lip} + \|\nabla V\|_{\rm Lip}) \|(\hat{x}-x)\psi\| \|(\hat{p}-p)\psi\| \\
&\leq (\|\nabla K\|_{\rm Lip} + \|\nabla V\|_{\rm Lip}) (\|(\hat{x}-x)\psi\|^2 + \|(\hat{p}-p)\psi\|^2) \\
&= (\|\nabla K\|_{\rm Lip} + \|\nabla V\|_{\rm Lip}) Q\,,
\end{align*}
as desired.

\subsection{Higher order moments}
\label{sec:Hgenloc}
Now we generalize the argument of the previous Section to more general Hamiltonians on phase space $\Real^{2D}$ and to higher order moments.

To work with $\hat{H}$ we first repurpose Lemma~\ref{lem:taylor} to apply to functions in $\mcal{S}_{\frac12}$.
\begin{lemma}
\label{lem:symbol-taylor}
Let $F\in C^\infty(\Real^N)$ with $\partial_j F\in\mcal{S}_{\frac12}$.
Then there exist functions $G_j\in \mcal{S}_{\frac12}$
such that
\[
F(x) = F(0) + \sum_{j=1}^N x_j G_j(x)\,.
\]
\end{lemma}
\begin{proof}
Again by the fundamental theorem of calculus we have
\[
G_j := \int_0^1 \partial_j F(sx)\diff s\,.
\]
The integrand is uniformly in $\mcal{S}_\frac12$ and the integral is over a bounded interval, so it follows that $G_j\in\mcal{S}_\frac12$.
\end{proof}

As a consequence of Lemma~\ref{lem:symbol-taylor} we have the following.
\begin{lemma}
Let $H$ be an admissible Hamiltonian as in Definition~\ref{def:admissibleH}.  Then
for any $\alpha\in\Real^{2D}$ we have the following ``Taylor expansion''
\begin{equation}
\label{eq:op-taylor}
-\frac{i}{\hbar}[\hat{\alpha}_a, \hat{H}] =
\left(\sum_b \omega_{ab} \partial_b H(\alpha) \right)\Id
+ \sum_b \hat{R}_{ab;\alpha} \cdot (\hat{\alpha}_b - \alpha_b\Id)\,,
\end{equation}
where $\hat{R}_{ab;\alpha} = \Op_\hbar(R_{ab;\alpha})$ with $R_{ab;\alpha}\in\mcal{S}_{\frac12}$.  In particular,
$\hat{R}_{ab;\alpha}$ satisfies the bounds~\eqref{eq:comm-symbol} as in Proposition~\ref{prp:comm-ests}.
\end{lemma}

Before stating our assumptions relating
$H$ and $\hat{H}$, we introduce some notation.  We first define the $k$-th order localization norms
around $\alpha\in\Real^{2D}$,
\[
\|\psi\|_{\Sobolev^k_\alpha} := \sum_{|J|=k} \|(\hat{\alpha}-\alpha)_J \psi\|_{L^2}.
\]
A first observation is that the $\Sobolev^k_\alpha$ norm is bounded in terms of the $\Sobolev^1_\alpha$ norm due to the uncertainty
principle.  In particular, the Heisenberg uncertainty principle
\[
\|(\hat{x}-x_0)\psi\|_{L^2}\|(\hat{p}-p_0)\psi\|_{L^2}  \geq \frac12 \,\hbar
\]
implies, by the AM-GM inequality, the bound
\[
\|(\hat{x}-x_0)\psi\|_{L^2} + \|(\hat{p}-p_0)\psi\|_{L^2} \geq \sqrt{\hbar}\,.
\]
Therefore, we have by induction the consequence
\begin{equation}
\label{eq:Hk-lb}
\|\psi\|_{\Sobolev^k_\alpha} \geq \sqrt{\hbar}\, \|\psi\|_{\Sobolev^{k-1}_\alpha}.
\end{equation}

Before we are ready to prove the bound~\eqref{eq:general-bd} we need one more lemma.
\begin{lemma}
\label{lemma:mixed-monomial-bd}
Let $\vec{m}\in\bbN^k$ and $\vec{m}'\in\bbN^{k'}$, $R\in\mcal{S}_{\frac12}$, and $\alpha\in\Real^{2D}$.  Then for
$\hat{R}=\Op_\hbar(R)$,
\begin{equation}
\label{eq:mixed-monomial-bd}
\|(\hat{\alpha}-\alpha)_{\vec{m}}\, \hat{R} \,(\hat{\alpha}-\alpha)_{\vec{m}'} \,\psi\|_{L^2}
\lsim \|\psi\|_{\Sobolev^{k+k'}_\alpha}.
\end{equation}
\end{lemma}
\begin{proof}
First we prove~\eqref{eq:mixed-monomial-bd} with $R=1$.  In this case we may take $k'=0$ without loss of generality.  The proof
is by induction on $k$, and the result clearly holds for $k=0,1$.  For $k>1$ we observe that monomials can be written in the form
\[
(\hat{\alpha}-\alpha)_{\vec{m}} = (\hat{\alpha}-\alpha)_{n(\vec{m})} + \sum_{j=1}^k\hbar^j P_{k-j}(\hat{\alpha}-\alpha),
\]
where $P_{j}$ are (non-commutative) homogeneous polynomials of degree $j$, with coefficients independent of $\hbar$.  Therefore by induction
we have
\begin{align*}
\|(\hat{\alpha}-\alpha)_{\vec{m}}\psi\|_{L^2}
\leq \|\psi\|_{\Sobolev^k_\alpha} + \sum_{j=1}^k \hbar^j \|\psi\|_{\Sobolev^{k-j}_\alpha}.
\end{align*}
The result now follows from the uncertainty principle in the form~\eqref{eq:Hk-lb}.

Now to prove~\eqref{eq:mixed-monomial-bd} in the general case, we observe that
\[
(\hat{\alpha}-\alpha)_{\vec{m}} \,\hat{R} = \sum_{S\subset[k]} \ad^{n(\vec{m}_S)}[\hat{R}] (\hat{\alpha}-\alpha)_{\vec{m}_{[k]\setminus S}}.
\]
Above $\vec{m}_S$ is the substring of $\vec{m}$ obtained by keeping only the entries specified by the set $S$.  Therefore,
and again using~\eqref{eq:Hk-lb} we have
\begin{align*}
\|(\hat{\alpha}-\alpha)_{\vec{m}} \hat{R}\,(\hat{\alpha}-\alpha)_{\vec{m}'}\psi\|_{L^2}
&\leq \sum_{S\subset[k]}
\|\ad^{n(\vec{m}_S)}[\hat{R}](\hat{\alpha}-\alpha)_{\vec{m}_{[k]\setminus S} + \vec{m}'}\psi\|_{L^2} \\
&\lsim \sum_{S\subset[k]} \hbar^{|S|/2} \|\psi\|_{\Sobolev^{k+k'-|S|}_\alpha} \\
&\lsim \|\psi\|_{\Sobolev^{k+k'}_\alpha}.
\end{align*}
\end{proof}

We are now ready to complete the proof of Equation~\eqref{eq:general-bd} of Proposition~\ref{prp:localization},
which follows from the estimate
\begin{equation}
\label{eq:time-derivative}
\frac{d}{dt} \|U(t)\psi\|_{\Sobolev^k_{\alpha(t)}}
\leq C_k \|U(t)\psi\|_{\Sobolev^{k}_{\alpha(t)}}\,,
\end{equation}
where as usual $U(t) = e^{- i H t/\hbar}$.

Using~\eqref{eq:op-taylor} we have for any $\alpha\in\Real^{2D}$
\[
\left.\frac{d}{dt} (\hat{\alpha}_a(t) - \alpha_a(t))\right|_{t=0}
= \sum_b \hat{R}_{ab;\alpha} (\hat{\alpha}_b-\alpha_b)\,.
\]

Therefore, for any multi-index $J$ we have that the derivative $\frac{d}{dt}(\hat{\alpha}-\alpha)_J$ is a finite sum of
terms of the form $(\hat{\alpha}-\alpha)_{\vec{m}_1} \,\hat{R}_{ab;\alpha}\, (\hat{\alpha}-\alpha)_{\vec{m}_2}$ for some monomial
strings $|\vec{m}_1| + |\vec{m}_2| = k$, and the result then follows from Lemma~\ref{lemma:mixed-monomial-bd}.

\subsection{Proof of Corollary~\ref{cor:from-mixture}}
\label{sec:pffrommixture}
For convenience of the reader we recall the statement of Corollary~\ref{cor:from-mixture}.
\begin{customcor}{1.5}
\label{cor:from-mixture2}
Let $p\in L^1(\Real^{2D})$ be a probability density function and
$\hat{\rho} = \int p(\alpha) \ket{\alpha}\bra{\alpha}\diff\alpha$ be a density matrix that is a mixture of coherent states according to $p$.
Then for general admissible Hamiltonians,
\begin{equation}
d_{W_p}(\Husimi_{\mcal{U}_T\rho}, \Phi_T \Husimi_{\rho}) \leq e^{C_pT} \hbar^{1/2}.
\end{equation}
If instead $\hat{H}=K(\hat{p}) + V(\hat{x})$ we have the bound
\begin{equation}
\frac{1}{2D}\,d_{W_2}(\Husimi_{\mcal{U}_T\rho}, \Phi_T \Husimi_{\rho})^2 \leq (1 + 2 e^{\Lambda_HT}) \hbar.
\end{equation}
\end{customcor}

\begin{proof}
Here we will use $W_q$ instead of $W_p$ to avoid confusion with $p(\alpha)$.  The $W_q$ distances satisfy the following convexity property:
\begin{align*}
d_{W_q}(\Husimi_{\mcal{U}_T\rho}, \Phi_T \Husimi_{\rho})^q \leq
\int p(\alpha) d_{W_q}(\Husimi_{\mcal{U}_T[\ket{\alpha}\bra{\alpha}]}, \Phi_T\Husimi_{\ket{\alpha}\bra{\alpha}})^q \diff \alpha.
\end{align*}
Let $\delta_\alpha$ be the Dirac delta measure located at $\alpha$.  Then we can bound the above
integrand by
\begin{align*}
d_{W_q}(\Husimi_{\mcal{U}_T[\ket{\alpha}\bra{\alpha}]}, \Phi_T\Husimi_{\ket{\alpha}\bra{\alpha}})
&\leq
d_{W_q}(\Husimi_{\mcal{U}_T[\ket{\alpha}\bra{\alpha}]}, \delta_{\Phi_T\alpha}) +
d_{W_q}(\Phi_T \Husimi_{\ket{\alpha}\bra{\alpha}}, \delta_{\Phi_T\alpha})\,.
\end{align*}
We treat each term above differently depending on the assumptions about $H$.

\textit{Case I: General Hamiltonians.}
Without loss of generality (by the monotonicity of $d_{W_q}$ in $q$) we can assume that $q=2k$ for some $k\in\bbN$.
Then we compute
\begin{align*}
\int |\beta-\alpha|^{2k}\Husimi_{\ket{\psi}\bra{\psi}}(\beta)\diff \beta
&= \int |\beta-\alpha|^{2k} |\Braket{\beta|\psi}|^2 \diff \beta \\
&\lsim_{k,D} \sum_{a=1}^{2D}
\int |\beta_a-\alpha_a|^{2k} |\Braket{\beta|\psi}|^2 \diff \beta \\
&\lsim_{k,D} \|\psi\|_{\Sobolev^k_\alpha}.
\end{align*}
Taking $\psi = e^{-i\hat{H}T/\hbar}\ket{\alpha}$, this proves
\begin{equation}
d_{W_q}(\Husimi_{\mcal{U}_T[\ket{\alpha}\bra{\alpha}]}, \delta_{\Phi_T \alpha})
\lsim e^{CT} \hbar^{1/2}.
\end{equation}
Moreover, $\Phi_T$ is $e^{CT}$-Lipschitz, so that
\begin{align*}
d_{W_q}(\Phi_T \Husimi_{\ket{\alpha}\bra{\alpha}}, \delta_{\Phi_T\alpha})
&\lsim e^{CT} d_{W_p}(\Husimi_{\ket{\alpha}\bra{\alpha}}, \delta_{\alpha})\\
&\lsim e^{CT}\hbar^{1/2}.
\end{align*}
Combining these bounds completes the proof of~\eqref{eq:mixture-general}.

\textit{Case II: Hamiltonians of the form $K(p)+V(x)$.}
An elementary calculation shows that in this case $\Phi_T$ is $e^{\Lambda_H T}$-Lipschitz\footnote{In fact, one can use the bound~\eqref{eq:meanfield-bd} in the limit $\hbar\to 0$ with $\psi = \ket{\beta}$ to see that $|\Phi_T \beta-\Phi_T \alpha| \leq e^{\Lambda_HT}|\beta-\alpha|$.} and that $d_{W_2}(\Husimi_{\ket{\alpha}\bra{\alpha}}, \delta_\alpha)^2 = 2D\hbar$.
It therefore suffices to estimate
\begin{align*}
\int |\beta-\alpha|^2\Husimi_{\ket{\psi}\bra{\psi}}(\beta)\diff \beta
&=
\sum_j
\int |\beta_a-\alpha_a|^2\Husimi_{\ket{\psi}\bra{\psi}}(\beta)\diff \beta \\
&=
\sum_j
\trace\!\left[ \Big(\int |\beta_a-\alpha_a|^2\ket{\beta}\bra{\beta}\diff \beta\Big)\ket{\psi}\bra{\psi}\right].
\end{align*}
The operator in brackets can alternatively be written as follows:
\begin{align*}
\int |\beta_a-\alpha_a|^2\ket{\beta}\bra{\beta}\diff \beta
&= \Op_\hbar( |\beta_a-\alpha_a|^2 \ast \gamma_\hbar) \\
&= \Op_\hbar(|\beta_a-\alpha_a|^2 +  \hbar) \\
&= (\hat{\beta}_a - \alpha_a)^2 + \hbar\,\Id\,.
\end{align*}
Therefore, rearranging and then using~\eqref{eq:meanfield-bd} yields
\begin{align*}
d_{W_2} (\Husimi_{\mcal{U}_T[\ket{\alpha}\bra{\alpha}]}, \delta_{\Phi_T\alpha})^2
&\leq \Big(\sum_j \|(\hat{\alpha}_a - (\Phi_T\alpha)_a) e^{-i\hat{H}T/\hbar}\ket{\alpha}\|_{L^2}^2 \Big)
+ 2D \hbar \\
&= \|e^{-i\hat{H}T/\hbar} \ket{\alpha}\|_{\Sobolev^1_{\Phi_T\alpha}}^2 + 2D\hbar \\
&\leq (1 + e^{2\Lambda_H T}) 2D\hbar\,,
\end{align*}
where we recall that $D = N d$ where $N$ is the number of particles and $d$ is the number of spatial dimensions.
\end{proof}

\section{Transport by localized unitaries}
\label{sec:W-transport}

In this Section we state a precise version of Proposition~\ref{prp:local-transport-informal} and prove it.
The first step is to quantify what it means for a local unitary to be well localized in phase space,
and to show that the operator
\begin{equation}
\label{eq:Wshift-def}
W_{T,\alpha} := e^{i\hat{H}T/\hbar} \tau_\alpha e^{-i\hat{H}T/\hbar}.
\end{equation}
is well-localized.  This is done in Section~\ref{sec:W-local}.  Then in Section~\ref{sec:mainthm} we state
Proposition~\ref{prp:local-transport-informal} and explain the proof strategy.  In Section~\ref{sec:phibds} we record some
elementary bounds, and in Sections~\ref{sec:proofM}-\ref{sec:proofE} the proof is completed.

\subsection{Local unitary operators}
\label{sec:W-local}
We measure the localization of a unitary operator in phase space using centered isotropic Sobolev norms applied
to coherent states.  To this end we define
\[
\|W\|_{Z^k} := \sup_\alpha \|W\ket{\alpha}\|_{\Sobolev^k_\alpha}.
\]
Note that $\|W\|_{Z^k} \geq \hbar^{1/2} \|W\|_{Z^{k-1}}$ by~\eqref{eq:Hk-lb}.
This norm bounds the off-diagonal matrix elements of the operator $W$.
\begin{lemma}
\label{lem:Z-offdiag}
For any operator $W$, and $\alpha,\beta\in\Real^{2D}$,
\[
|\Braket{\alpha|W|\beta}|
\leq C_{k,D} |\alpha-\beta|^{-k} \|W\|_{Z^k}.
\]
\end{lemma}
\begin{proof}
Without loss of generality we can suppose that $|\alpha-\beta| \leq C_{D} |\alpha_{1}-\beta_{1}|$.
Define $f(x) = (\hbar + (x - \beta_{1})^2)^{k/2}$ and $g(x) = \frac{1}{f(x)}$.  Then
\begin{align*}
|\Braket{\alpha|W|\beta}|
&= |\Braket{\alpha| g(\hat{x}_1) f(\hat{x}_1) W |\beta}| \\
&\leq
\| (\hbar + (\hat{x}_1-\beta_{1})^2)^{-k/2} \ket{\alpha}\|_{L^2}
\| (\hbar + (\hat{x}_1-\beta_{1})^2)^{k/2} W\ket{\beta}\|_{L^2} \\
&\leq |\alpha_{1}-\beta_{1}|^{-k}\|W\|_{Z^k} \\
&\leq C_{k,D} |\alpha-\beta|^{-k} \|W\|_{Z^k}.
\end{align*}
\end{proof}

Our goal is to bound $\|W_{T,\alpha}\|_{Z^k}$, where $W_{T,\alpha} = e^{i\hat{H}T/\hbar} \tau_\alpha e^{-i\hat{H}T/\hbar}$ involves a phase space shift operator $\tau_\alpha$.  We therefore need to know how the isotropic Sobolev norms are affected under such shifts, and this is recorded in the following lemma:
\begin{lemma}
\label{lem:sobolev-shift}
Let $\psi\in \Sobolev^k$ and $\alpha,\beta\in\Real^{2D}$.  Then
\[
\|\tau_\beta\psi\|_{\Sobolev^k_\alpha} \leq C_k (1 + |\beta|\hbar^{-1/2})^k \|\psi\|_{\Sobolev^k_\alpha}.
\]
\end{lemma}
\begin{proof}
The result for $k=0$ follows from the boundedness of the shift operator
(since $\|\cdot\|_{\Sobolev^0_\alpha} = \|\cdot\|_{L^2}$.  Now suppose $k\geq 1$, and let $\vec{m}\in [2D]^k$ be a tuple
specifying some monomial of degree $k$.
We use the identity $\hat{\alpha}_j \tau_\beta = \tau_\beta \hat{\alpha}_j + \beta_j \tau_\beta$
and write $\vec{m}'\in[2D]^{k-1}$ for the $k-1$ subtuple of $\vec{m}$
to estimate
\begin{align*}
\| (\hat{\alpha}-\alpha)_{\vec{m}'} \tau_\beta \psi\|_{L^2}
&\leq
\|(\hat{\alpha}-\alpha)_{\vec{m}'} \tau_\beta (\hat{\alpha}_{m_k}-\alpha_{m_k})\psi\|_{L^2}
+ |\beta_{m_k}| \|(\hat{\alpha}-\alpha)^{\vec{m}'} \tau_\beta\psi\|_{L^2} \\
&\leq \|\tau_\beta(\hat{\alpha}-\alpha) \psi\|_{\Sobolev^{k-1}_\alpha}
+  |\beta_{m_k}| \|\tau_\beta \psi\|_{\Sobolev^{k-1}_\alpha} \\
&\lsim (1+ |\beta|\hbar^{-1/2})^{k-1} \|(\hat{\alpha}-\alpha)\psi\|_{\Sobolev^{k-1}_\alpha}
+  |\beta_{m_k}| ( 1 + |\beta|\hbar^{-1/2})^{k-1} \|\psi\|_{\Sobolev^{k-1}_\alpha} \\
&\lsim (1+ |\beta|\hbar^{-1/2})^{k-1} \|\psi\|_{\Sobolev^{k}_\alpha}
+  \hbar^{-1/2}|\beta_{m_k}| ( 1 + |\beta|\hbar^{-1/2})^{k-1} \|\psi\|_{\Sobolev^{k}_\alpha} \\
&\lsim (1 + |\beta| \hbar^{-1/2})^k \|\psi\|_{\Sobolev^k_\alpha}.
\end{align*}
Above we have used the bound $\|\psi\|_{\Sobolev^k_\alpha} \leq \hbar^{-1/2} \|\psi\|_{\Sobolev^{k-1}_\alpha}$
(see~\eqref{eq:Hk-lb}) and the inductive hypothesis.
\end{proof}

We can now use Proposition~\ref{prp:localization} to bound the phase-space localization of $W_{T,\alpha}$.
\begin{lemma}
Let $H$ be an admissible Hamiltonian in the sense that $\partial_{ab} H\in\mcal{S}_{\frac12}$ for all $a,b\in [2D]$.
Then for any $T\in\Real$ and $\alpha\in\Real^{2D}$ the operator $W_{T,\alpha}$ defined in~\eqref{eq:Wshift-def}
satisfies
\[
\|W_{T,\alpha}\|_{Z^k} \leq e^{C_k T}(\hbar^{k/2} + |\alpha|^k)\,.
\]
\end{lemma}
\begin{proof}
Fix some $\alpha_0 \in\Real^{2D}$ and let $\alpha(T)$ be the classical evolution at time $T$.
We use Proposition~\ref{prp:localization} and Lemma~\ref{lem:sobolev-shift} to estimate
\begin{align*}
\|e^{i\hat{H}T/\hbar} \tau_\alpha e^{-i\hat{H}T/\hbar}\ket{\alpha_0}\|_{\Sobolev^k_{\alpha_0}}
&\lsim e^{C_kT} \|\tau_\alpha e^{-i\hat{H}T/\hbar}\ket{\alpha_0}\|_{\Sobolev^k_{\alpha(T)}} \\
&\lsim e^{C_kT}(1 + |\alpha|^k\hbar^{-k/2}) \|e^{-i\hat{H}T/\hbar}\ket{\alpha_0}\|_{\Sobolev^k_{\alpha(T)}} \\
&\lsim e^{C_kT}(1 + |\alpha|^k\hbar^{-k/2}) \|\ket{\alpha_0}\|_{\Sobolev^k_{\alpha_0}} \\
&\lsim e^{C_kT}(1 + |\alpha|^k\hbar^{-k/2}) \hbar^{k/2}.
\end{align*}
\end{proof}

\subsection{Bound on transport by local unitaries}
\label{sec:mainthm}
We are now ready to state the main result concerning the transport distance between $\rho$ and $W\rho W^\dagger$.
\begin{theorem}
\label{thm:main-husimi}
Let $r>\hbar^{1/2}$ and $W$ be a unitary operator satisfying $\|W\|_{Z^k} <\infty$ for all $k\in\bbN$, and $1\leq p<\infty$.   Then there is some $M=M(p,D)$ such that
\[
d_{W_p}(\Husimi_\rho, \Husimi_{W\rho W^\dagger}) \leq Cr (\hbar^{-2D} r^{4D})(1 + r^{-M}\|W\|_{Z^M})^2.
\]
In particular, taking $W=e^{i\hat{H}T/\hbar} \tau_\alpha e^{-i\hat{H}T/\hbar}$ for an admissible Hamiltonian
$\hat{H}$, we have, taking $r = \hbar^{1/2} + |\alpha|$,
\begin{equation}
\label{eq:husimi-shift}
d_{W_p}(\Husimi_\rho, \Husimi_{W\rho W^\dagger}) \leq Ce^{C_M(H) T}
(\hbar^{1/2} + |\alpha|)(1 + |\alpha|^{4D} \hbar^{-2D}).
\end{equation}
\end{theorem}
The term involving $\hbar^{-2D}r^{4D}$ is possibly unnecessary, and represents some inefficiency
in the argument.  Moreover, one might expect that one can take $M$ independent of $D$, whereas
the proof given below can be made to work only if $M>4D+p+2$.
The quantity $1 + r^{-M} \|W\|_{Z^M}$ appears repeatedly throughout the proof, so we use the more convenient
notation
\[
\|W\|_{Z^M_r} := 1 + r^{-M} \|W\|_{Z^M}.
\]

The proof of Theorem~\ref{thm:main-husimi} goes through Kantorovich duality, which gives
the identity
\[
d_{W_p}(\mu,\nu)^p = \sup_{f(x)+g(y)\leq |x-y|^p} \int f\diff \mu + \int g\diff \nu.
\]
In the supremum above, no regularity is imposed on $f$ and $g$ aside from measurability.
On the other hand, when we take $\mu = \Husimi_\rho$ and $\nu=\Husimi_{W\rho W^\dagger}$, it will
be helpful to have some regularity for $f$ and $g$.   This can be handled by performing a
mild mollification.
\begin{lemma}
Let $\mu$ and $\nu$ be probability measures and let $\kappa$ be some convolution kernel satisfying
$\int |x|^p\kappa(x)\diff x := r_p(\kappa)^p < \infty$.  Then we have
\[
d_{W_p}(\mu,\nu)^p \leq r_p(\kappa)^p +
\sup_{f(x)+g(y)\leq |x-y|^p} \int f\ast \kappa \diff \mu + \int g\ast \kappa \diff \nu\,.
\]
\end{lemma}
\begin{proof}
We have $d_{W_p}(\mu,\mu\ast\kappa) \leq r_p(\kappa)$, so by the triangle inequality
\[
d_{W_p}(\mu,\nu) \leq d_{W_p}(\mu\ast \kappa, \nu\ast \kappa) + r_p(\kappa).
\]
Then we apply Kantorovich duality to the first term.
\end{proof}

We will need to use a specific convolution kernel with polynomial decay.
To this end, we fix $K = 2D+10+2\lceil p\rceil$ and define
$\varphi(\alpha) := C(1+|\alpha|^2)^{-K/2}$, with $C=C_K$ chosen such that $\int \varphi = 1$.  The exponent is chosen to be large enough that $\int |\alpha|^p \varphi(\alpha) < \infty$.
We use also the scaled version $\varphi_R(x) := R^{-D} \varphi(x/R)$, where we take
$R=Q\ell$ with $Q$ an as-yet unspecified large constant and $\ell$ given by
\begin{equation}
\label{eq:ell-def}
\ell := r (\hbar^{-2D}r^{4D}) \|W\|_{Z^{K+2D+1}_r}^2.
\end{equation}

Theorem~\ref{thm:main-husimi} will then follow if we show that for arbitrary measurable
$f,g$ satisfying $f(\alpha)+g(\beta) \leq |\alpha-\beta|^p$,
\begin{equation}
\label{eq:husimi-KR}
\int (f\ast \varphi_{R})(\alpha) \Husimi_\rho(\alpha)\diff\alpha
+ \int (g\ast\varphi_{R})(\alpha) \Husimi_{W\rho W^\dagger}(\alpha)\diff\alpha
\leq Q R^p\,.
\end{equation}
For convenience we define $\delta := QR^p$ to be the term on the right-hand side.

Using the wavepacket quantization
\begin{align}
\label{E:Opwdef1}
\Op^w(f) := (2\pi\hbar)^{-D} \int f(\alpha) \,|\alpha\rangle \langle \alpha|\,\diff\alpha,
\end{align}
we have the identity
\[
\int \Husimi_\rho(\alpha) f(\alpha)\diff\alpha = \trace[\rho \,\Op^w(f)]\,.
\]

Therefore, the bound~\eqref{eq:husimi-KR} is equivalent to the bound
\[
\trace[\rho (\Op^w(f\ast \varphi_{R}) + W \Op^w(g\ast \varphi_{R}) W^\dagger) ]
\leq \delta.
\]
Since $\rho$ is arbitrary, proving the above bound requires the following operator inequality:
\[
\Op^w(f\ast \varphi_{R}) + W\Op^w(g\ast\varphi_{R}) W^\dagger
\preceq \delta\,\text{Id}\,.
\]

Now define $D(\alpha) := -(f(\alpha)+g(\alpha))$, which is a nonnegative function because
$f(\alpha)+g(\alpha) \leq |\alpha-\alpha|^p = 0$.  We can then rearrange the above inequality into
\begin{equation}
\label{eq:op-positivity}
0 \preceq \Op^w(D\ast\varphi_{R}) + \delta + [\Op^w(g\ast\varphi_{R}) - W\Op^w(g\ast\varphi_{R})W^\dagger].
\end{equation}

One can think of the error term
$\Op^w(g\ast \varphi_{R}) - W\Op^w(g\ast \varphi_R) W^\dagger$ on the right as expressing how much $g$ changes when
it is wiggled slightly by $W$.  The point
of the inequality above is that the function $D$ controls the regularity of $g$, acting as a sort of modulus of continuity.  Indeed, adding $D(\alpha)$ to both sides of
\[
f(\alpha)+g(\beta) \leq |\alpha-\beta|^p
\]
yields
\[
g(\beta) - g(\alpha) \leq |\alpha-\beta|^p + D(\alpha).
\]
By symmetry, we therefore have
\begin{equation}
\label{eq:gd-cty}
|g(\alpha)-g(\beta)| \leq |\alpha-\beta|^p + D(\alpha)+D(\beta).
\end{equation}
This inequality justifies the notion that $D$ acts as a sort of ``modulus of continuity'' for $g$.   

Although the reasoning above gives some heuristic justification for the operator inequality~\eqref{eq:op-positivity}, it is unclear however how one would establish
positivity for this operator by direct means.  
For this we use the following abstract
result.
\begin{lemma}
\label{lem:abstract-positivity}
Let $A \in\mcal{B}(\mcal{H})$ be a bounded self-adjoint operator on a Hilbert space $\mcal{H}$.  Suppose
that $A = M+E$ where $M \succeq \eps\,\Id$ for some $\eps>0$ and $E$ is compact.  Moreover suppose that there exists a Banach space $X$ and a continuous inclusion $\iota: \mcal{H}\to X$ such that $M^{-1}E$ is a contraction on $X$.
 Then $A$ is positive semi-definite.
\end{lemma}
\begin{proof}
First note that the inner product $\langle \cdot,\cdot\rangle_M$ defined by
\[
\langle f,g\rangle_M :=
\langle f, Mg\rangle
\]
is equivalent to the standard inner product on $\mcal{H}$ because the operator $M$ is a continuous bijection on $\mcal{H}$ with inverse bounded by $\eps^{-1}$.  The operator $B := M^{-1}E$ is self-adjoint with respect to $\langle\cdot,\cdot\rangle_M$, and is also
compact because $E$ is compact.
The positivity of $A$ is then equivalent to the positivity of $\Id+B$
with respect to $\langle \cdot,\cdot\rangle_M$.  This follows if the spectrum $\sigma(B)$ satisfies $\sigma(B)\subset [-1,1]$.

Since $B$ is compact, we
simply need to show that every eigenvalue
$\lambda_k$ of $B$ satisfies $|\lambda_k|\leq 1$.
Let $\psi_k\in\mcal{H}$ be an eigenvector of $B$ with eigenvalue $\lambda_k$, so that $B \psi_k= \lambda_k\psi_k$.
Normalize $\psi_k$ so that $\|\psi_k\|_X=1$. Now we use that $B$ is a contraction on $X$ to see that
\[
|\lambda_k|
= \|B_\eps \psi_k\|_X \leq \|\psi_k\|_X = 1,
\]
as desired.
\end{proof}

To establish~\eqref{eq:op-positivity} we invoke Lemma~\ref{lem:abstract-positivity} with
\begin{align*}
M &:= \Op^w(D\ast \varphi_{R})+\delta \\
E &:= \Op^w(g\ast \varphi_R) - W \Op^w(g\ast \varphi_R) W^\dagger.
\end{align*}
The strict positivity of $M$ follows because $\delta>0$ and $D \succeq 0$.
Moreover, we can have compactness of $E$ if we additionally require that $g$ is compactly supported
(which of course can be done without changing the supremum in the Kantorovich duality).

It therefore suffices to find some norm $X$ with an inclusion $\mcal{H}\hookrightarrow X$
for which $M^{-1}E$ is a contraction.  This norm can be explicitly defined:
\[
\|\psi\|_X := \sup_\alpha |\Braket{\alpha|\psi}|.
\]
In our proof we also use the auxiliary norm
\[
\|\psi\|_Y := \sup_\alpha (D\ast \varphi_{R}(\alpha) +\delta)^{-1} |\Braket{\alpha|\psi}|\,.
\]

The main result Theorem~\ref{thm:main-husimi} then follows from the pair of lemmas:
\begin{lemma}
\label{lem:MYX}
Let $M$ be as above (including all implicit assumptions about $f,g,D$, etc.).  Then
\begin{equation}
\label{eq:Minv}
\|M^{-1}\|_{Y\to X} \leq C.
\end{equation}
\end{lemma}

\begin{lemma}
\label{lem:EXY}
Let $E$ be as above (including all implicit assumptions about $f,g,D$, etc.).  Then
\begin{equation}
\label{eq:E}
\|E\|_{X\to Y} \leq CQ^{-1}.
\end{equation}
\end{lemma}
Above recall that $R = Q\ell$.  Then combining Lemma~\ref{lem:MYX} and
Lemma~\ref{lem:EXY} we have
\[
\|M^{-1}E\|_{X\to X} \leq C Q^{-1},
\]
so that choosing $Q$ large enough we have that $M^{-1}E$ is a contraction.  This implies
the positivity of the operator~\eqref{eq:op-positivity} by Lemma~\ref{lem:abstract-positivity}, and therefore the desired bound~\eqref{eq:husimi-KR} follows and with it the proof of Theorem~\ref{thm:main-husimi}.

The remainder of this Section is dedicated to the proofs of Lemmas~\ref{lem:MYX} and~\ref{lem:EXY}.
Before we can proceed with their proofs, we first need to take a step back to derive some
elementary facts about the convolution kernel $\varphi$.

\subsection{Elementary lemmas concerning the convolution kernel $\varphi$}
\label{sec:phibds}

In this Section we collect some elementary bounds involving $\varphi_R$ which will be used in the proofs of
Lemma~\ref{lem:MYX} and Lemma~\ref{lem:EXY} below.

Recall that
\[
\varphi(\alpha) := C_K (1 + |\alpha|^2)^{-K/2},
\]
with $C_K$ chosen so that $\int \varphi = 1$.  The actual exponent $K > 10D$ is kept implicit throughout (so constants will depend on $K$).  We will work with the scaled function $\varphi_R(x) := R^{-D} \varphi(x/R)$.

We collect some pointwise bounds on $\varphi_R$ below.

\begin{lemma}
\label{lem:phi-ptwise}
The convolution kernel $\varphi_R$ satisfies the following estimates for any $\alpha,\beta\in\Real^{2D}$:
\begin{align}
\label{eq:phi-gradient}
|\nabla \varphi_R(\alpha)| &\leq C R^{-1} (1+|\alpha/R|)^{-1} \varphi_R(\alpha) \\
\label{eq:growth}
|\varphi_R(\alpha+\beta)| &\leq C (1+|\beta/R|^K) \varphi_R(\alpha) \\
\label{eq:displacement}
|\varphi_R(\alpha+\beta) - \varphi_R(\alpha)| &\leq C R^{-1} |\beta| (1+|\beta/R|^K) \varphi_R(\alpha)\,.
\end{align}
\end{lemma}
\begin{proof}
The bound~\eqref{eq:phi-gradient} follows from rescaling the elementary inequality
\[
|\nabla \varphi(\alpha)| \leq C (1+|\alpha|)^{-1} \varphi(\alpha)\,.
\]
The bound~\eqref{eq:growth} simply follows because the right hand side is bounded from below by a constant, and the left hand side is bounded by a constant.

To prove~\eqref{eq:displacement} we use the fundamental theorem of calculus:
\[
|\varphi_R(\alpha+\beta) - \varphi_R(\alpha)| \leq |\beta|\int_0^1 |(\nabla \varphi_R)(\alpha+t\beta)|\diff t\,.
\]
To bound the integrand we first use~\eqref{eq:phi-gradient} and then~\eqref{eq:growth}
\begin{align*}
|\nabla\varphi_R(\alpha+\beta')|
&\leq C R^{-1} (1+|(\alpha+\beta')/R|)^{-1} \varphi_R(\alpha+\beta') \\
&\leq C R^{-1} (1+|(\alpha+\beta')/R|)^{-1} (1+|\beta'/R|^K) \varphi_R(\alpha)\,.
\end{align*}
We can drop the term $(1+|(\alpha+\beta')/R|)^{-1}$ which is bounded by $1$, and then integrate to conclude
the proof of~\eqref{eq:displacement}.
\end{proof}

The pointwise bounds in Lemma~\ref{lem:phi-ptwise} can then be used to prove bounds for
functions of the form $f_R = f\ast \varphi_R$ where $f\geq 0$ is a positive measurable function.
\begin{lemma}
\label{lem:phi-conv}
Let $\gamma$ be a positive convolution kernel satisfying $\int |\alpha|\gamma(\alpha)\diff \alpha = M_1<\infty$
and $\int |\alpha|^{K+1} \gamma(\alpha)\diff \alpha = M_{K+1} <\infty$ and $f\geq 0$ be a positive measurable function.
Then
\[
\int |f_R(\alpha) - f_R(\alpha+\beta)|\gamma(\beta)\diff \beta \leq C R^{-1}(M_1 + R^{-K} M_{K+1})f_R(\alpha)\,.
\]
\end{lemma}
\begin{proof}
Write out the difference $f_R(\alpha)-f_R(\alpha+\beta)$ as an integral and then apply~\eqref{eq:displacement}
to obtain
\begin{align*}
|f_R(\alpha) - f_R(\alpha+\beta)|
&\leq \int f(\zeta) |\varphi_R(\alpha-\zeta) - \varphi_R(\alpha+\beta-\zeta)| \diff \zeta\\
&\leq CR^{-1}|\beta| (1+|\beta/R|^K) f_R(\alpha)\,.
\end{align*}
\end{proof}

\subsection{Proof of Lemma~\ref{lem:MYX}}
\label{sec:proofM}
For simplicity we set $\widetilde{D} := D\ast\varphi_{R} + \delta$, which implicitly depends on the
parameters $R=Q\ell$ with $\ell > r$ defined in~\eqref{eq:ell-def}
and $\delta = QR^p$.  With this notation, we have that $M = \Op^w(\widetilde{D})$.

To work with $M^{-1}$ we first observe the identity
\begin{align*}
\Op^w(\widetilde{D}^{-1})\Op^w(\widetilde{D})
&= (2\pi\hbar)^{-2D}\int \widetilde{D}^{-1}(\alpha) \widetilde{D}(\beta) \ket{\alpha}\braket{\alpha|\beta}\bra{\beta} \diff \alpha\diff\beta \\
&= \Id + (2\pi \hbar)^{-2D} \int \widetilde{D}^{-1}(\alpha) (\widetilde{D}(\beta)-\widetilde{D}(\alpha)) \braket{\alpha|\beta} \ket{\alpha}\bra{\beta}\diff\alpha\diff\beta \\
&=: \Id + T\,.
\end{align*}
Rearranging, this gives
\[
\Op^w(\widetilde{D})^{-1} = (\Id+T)^{-1}  \Op^w(\widetilde{D}^{-1})\,.
\]
To prove~\eqref{eq:Minv} it suffices to prove both
\begin{equation}
\label{eq:T-contract}
\|T\|_{X\to X} \leq \frac12
\end{equation}
and
\begin{equation}
\label{eq:hinv-bd}
\|\Op^w(\widetilde{D}^{-1})\|_{Y\to X} \leq C\,.
\end{equation}

First we prove~\eqref{eq:T-contract}.  We have the bounds
\begin{align*}
\|T\|_{X\to X}
&\lsim (2\pi\hbar)^{-D} \sup_\xi \int |\widetilde{D}^{-1}(\xi) (\widetilde{D}(\eta)-\widetilde{D}(\xi))| |\Braket{\xi|\eta}|\diff\eta \\
&\lsim C Q^{-1} (\hbar^{1/2} \ell^{-1})\,.
\end{align*}
To get to the last line, we applied
Lemma~\ref{lem:phi-conv} to the right hand side with convolution kernel
$\gamma(\beta) := (2\pi\hbar)^{-D} |\Braket{\alpha|\alpha+\beta}|$ and $f = D + \delta$.  Thus
by choosing $Q$ large enough we have~\eqref{eq:T-contract}, as desired (since $\ell>\hbar^{1/2}$).

The bound~\eqref{eq:hinv-bd} almost follows by definition of $\|\cdot\|_Y$:
\begin{align*}
\|\Op^w(\widetilde{D}^{-1})\psi\|_X
&= \sup_\alpha |\Braket{\alpha|\Op^w(\widetilde{D}^{-1})\psi}| \\
&= \sup_\alpha (2\pi\hbar)^{-D} \int \widetilde{D}^{-1}(\beta)
|\Braket{\alpha|\beta}\,\Braket{\beta|\psi}| \diff \beta \\
&\leq \|\psi\|_Y \,\sup_\alpha \int (2\pi\hbar)^{-D} |\Braket{\alpha|\beta}|\diff\beta \\
&\leq C \|\psi\|_Y.
\end{align*}

\subsection{Proof of Lemma~\ref{lem:EXY}}
\label{sec:proofE}
Recall that
\[
E = \Op^w(g\ast \varphi_{R}) - W \Op^w(g\ast \varphi_{R}) W^\dagger.
\]
For simplicity we write $g_R := g\ast \varphi_{R}$.

We compute
\begin{align*}
\|E\psi\|_Y
&= \sup_\xi \widetilde{D}(\xi)^{-1} |\Braket{\xi|E\psi}| \\
&\leq \sup_\xi \widetilde{D}(\xi)^{-1}
\Big|(2\pi \hbar)^{-D}\int g_R(\alpha)\braket{\xi|\alpha}\braket{\alpha|\psi}\diff \alpha
- (2\pi\hbar)^{-D} \int g_R(\alpha) \braket{\xi|W|\alpha}\braket{\alpha|W^\dagger|\psi}\diff\alpha\Big|.
\end{align*}
Note that, by the identity $\Id = (2\pi\hbar)^{-D} \int \ket{\alpha}\bra{\alpha}\diff\alpha$ and the unitarity of $W$ we have
\[
g_R(\xi) \braket{\xi|\psi} = (2\pi\hbar)^{-D} \int g_R(\xi) \braket{\xi|\alpha}\braket{\alpha|\psi}\diff\alpha
= (2\pi\hbar)^{-D} \int g_R(\xi) \braket{\xi|W|\alpha}\braket{\alpha|W^\dagger|\psi}\diff \alpha\,,
\]
so that
\begin{align*}
\|E\psi\|_Y &\lsim
\sup_\xi
\widetilde{D}(\xi)^{-1} \hbar^{-D}
\Big[\int |g_R(\alpha)-g_R(\xi)| |\!\braket{\xi|\alpha}\braket{\alpha|\psi}\!|\diff \alpha
\nonumber \\
& \qquad \qquad \qquad \qquad \qquad \qquad + \int|g_R(\alpha)-g_R(\xi)| |\!\braket{\xi|W|\alpha}\braket{\alpha|W^\dagger|\psi}\!| \diff \alpha\Big].
\end{align*}
We will prove a bound on the second term, the first term having the same form but with $W = \Id$
and $\|\Id\|_{Z^M} \leq \|W\|_{Z^M}$.
First we use Lemma~\ref{lem:Z-offdiag} and a decomposition of the identity to estimate
\begin{align*}
|\Braket{\alpha |W^\dagger|\psi}|
&\leq (2\pi\hbar)^{-D} \int |\Braket{\alpha|W^\dagger|\beta}| |\Braket{\beta|\psi}| \diff \beta \\
&\leq \|\psi\|_X \|W\|_{Z^{2D+1}_r}
(2\pi\hbar)^{-D} \int (1 + r^{-1}|\alpha-\beta|)^{-2D-1} \diff \beta \\
&\lsim (\hbar^{-1/2}r)^{2D} \|\psi\|_X \|W\|_{Z^{2D+1}_r}.
\end{align*}

So to prove~\eqref{eq:E} it suffices to show that
\begin{equation}
\label{eq:integral-goal}
(2\pi\hbar)^{-D} \int |g_R(\alpha) -g_R(\xi)| |\Braket{\xi|W|\alpha}| \diff \alpha
\leq C Q^{-1} (\hbar^{1/2}r^{-1})^{2D} \|W\|_{Z^{2D+1}_r}^{-1} \widetilde{D}(\xi)\,.
\end{equation}
We have again by Lemma~\ref{lem:Z-offdiag} that
\[
|\Braket{\xi|W|\alpha}| \lsim r^{2D} \chi_r(\xi-\alpha) \|W\|_{Z^{2K}_r},
\]
where we have defined $\chi_r(\xi) := r^{-2D} (1+|\xi|/r)^{-K-2D-1}$ to be the same kind of
polynomial weight as $\varphi_r$, but with a higher order.
Upon changing variables, the integral we need to bound has the form
\[
(\hbar^{-1/2}r)^{2D} \int |g_R(\xi+\alpha) -g_R(\xi)| \chi_r(\alpha) \diff \alpha\,.
\]
Now we bound the difference in the first term by writing out the convolution:
\begin{align*}
|g_R(\xi+\alpha)-g_R(\xi)|
&= \Big| \int g(\beta) (\varphi_R(\xi+\alpha-\beta) - \varphi_R(\xi-\beta))\diff \beta\Big| \\
&= \Big| \int (g(\beta) - g(\xi+(R/r)\alpha)) (\varphi_R(\xi+\alpha-\beta) - \varphi_R(\xi-\beta))
\diff \beta\Big|  \\
&\leq \int |g(\beta) - g(\xi+(R/r)\alpha)| |\varphi_R(\xi+\alpha-\beta) - \varphi_R(\xi-\beta)|
\diff \beta \\
&\leq \int \Big(|\beta-\xi-(R/r)\alpha|^p + D(\beta) + D(\xi+(R/r)\alpha)\Big) \\
& \qquad \qquad \qquad \qquad \qquad \quad \, \times |\varphi_R(\xi+\alpha-\beta) - \varphi_R(\xi-\beta)| \diff \beta\,. 
\end{align*}
We therefore have
\[
\int |g_R(\xi+\alpha) -g_R(\xi)| \chi_r(\alpha) \diff \alpha \leq {\rm I} +
{\rm II} + {\rm III}\,,
\]
where
\begin{align}
\label{eq:I}
{\rm I} &:= \int \int|\beta-\xi-(R/r)\alpha|^p
|\varphi_R(\xi+\alpha-\beta) - \varphi_R(\xi-\beta)| \chi_r(\alpha)\diff \beta \diff \alpha
\\
\label{eq:II}
{\rm II} &:= \int \int D(\beta)
|\varphi_R(\xi+\alpha-\beta) - \varphi_R(\xi-\beta)| \chi_r(\alpha)\diff \beta \diff \alpha\\
\label{eq:III}
{\rm III} &:= \int \int D(\xi+(R/r)\alpha)
|\varphi_R(\xi+\alpha-\beta) - \varphi_R(\xi-\beta)| \chi_r(\alpha)\diff \beta \diff \alpha.
\end{align}

To bound the first integral~\eqref{eq:I} above we used the inequality~\eqref{eq:displacement}, which appropriately rescaled becomes
\begin{equation}
\label{eq:R-displacement}
|\varphi_R(\xi +\alpha-\beta) - \varphi_R(\xi-\beta)|
\leq C R^{-1} |\alpha| (1+|\alpha/R|^K) \varphi_R(\xi-\beta)\,.
\end{equation}
Therefore, using also $|\beta-\xi-t\alpha|^p \leq C_p( |\beta-\xi|^p + t^p|\alpha|^p)$, we
have \begin{align*}
{\rm I} &\leq
C R^{-1} \int \Big(\int |\beta-\xi|^p \varphi_R(\xi-\beta) \diff \beta\Big)|\alpha| (1+|\alpha/R|^K)
\chi_r(\alpha)\diff \alpha \\
&\qquad +
C R^{-1} \int \Big(\int \varphi_R(\xi-\beta) \diff \beta\Big)(R/r)^p |\alpha|^p |\alpha| (1+|\alpha/R|^K)
\chi_r(\alpha)\diff \alpha \\
&\lsim C(r/R) R^p.
\end{align*}

For the second term~\eqref{eq:II} we use~\eqref{eq:R-displacement} again:
\begin{align*}
{\rm II} &\leq C R^{-1} \int D(\beta) \varphi_R(\xi-\beta) \diff \beta
\int
|\alpha| (1+|\alpha/R|^K) \chi_r(\alpha)
\diff \alpha \\
&\leq C (D\ast \varphi_R)(\xi) (r/R)\,.
\end{align*}

For the third term~\eqref{eq:III} we note that the integral over $\beta$ can be done first, and produces a factor of $R^{-1}|\alpha|$ so that
\[
{\rm III} \leq \int D(\xi+(R/r)\alpha) R^{-1} |\alpha| \chi_r(\alpha) \diff\alpha\,.
\]
Now relabeling $\beta =(R/r)\alpha$ this is bounded by
\begin{align}
{\rm III}
&\leq (rR^{-1}) \int D(\xi+\beta) (R^{-1} |\beta|) \chi_R(\beta) \diff\beta \\
&\lsim (rR^{-1}) (D\ast \varphi_R)(\xi)\,.
\end{align}

Combining our bounds for~\eqref{eq:I},~\eqref{eq:II}, and~\eqref{eq:III} we obtain
\[
\int |g_R(\alpha)-g_R(\xi)| |\Braket{\xi|W|\alpha}| \diff \alpha
\leq C (r/R) (\hbar^{-1/2}r)^{2D} \|W\|_{Z^M_r} (D\ast \varphi_R(\xi) + R^p)\,,
\]
and this implies~\eqref{eq:integral-goal} by the definition of $R=Q\ell$ with
$\ell = r \hbar^{-2D}r^{4D} \|W\|_{Z^{2K}_r}^2$.

\section{Proofs of main results}
\label{sec:proof}

In this Section we complete the proofs of the main results.

\subsection{Proof of Theorem~\ref{thm:Wp-egorov}}

Following the argument outlined in Section~\ref{sec:outline}, we use the triangle inequality to bound
\begin{align*}
d_{W_p} (\Husimi_{\mcal{U}_T\rho}, \Phi_T \Husimi_\rho)
\leq
d_{W_p} (\Husimi_{\mcal{U}_T\rho}, \Husimi_{\mcal{U}_T\circ\mcal{N}[\rho]})
+
d_{W_p} (\Husimi_{\mcal{U}_T\circ\mcal{N}[\rho]}, \Phi_T \Husimi_{\mcal{N}[\rho]})
+
d_{W_p} (\Phi_T \Husimi_{\mcal{N}[\rho]}, \Phi_T \Husimi_\rho)\,,
\end{align*}
where $\mcal{N}$ is the noising superoperator.  By Corollary~\ref{cor:from-mixture} (proved in Section~\ref{sec:pffrommixture}) we have
\[
d_{W_p} (\Husimi_{\mcal{U}_T\circ\mcal{N}[\rho]}, \Phi_T \Husimi_{\mcal{N}[\rho]})
\leq e^{C_pT} \hbar^{1/2}.
\]
Then because $\Phi_T$ is $e^{CT}$-Lipschitz we have
\[
d_{W_p} (\Phi_T \Husimi_{\mcal{N}[\rho]}, \Phi_T \Husimi_\rho)
\leq e^{CT} d_{W_p} (\Husimi_{\mcal{N}[\rho]}, \Husimi_\rho)
\leq C_p e^{CT}\hbar^{1/2},
\]
the second inequality following as a consequence of the identity
\[
\Husimi_{\mcal{N}[\rho]} =  \Husimi_\rho \ast \gamma_{2\hbar}\,.
\]
For the remaining term we use the identity
\[
\mcal{N}[\rho] = \int \mcal{T}_\alpha[\rho]\, \gamma_\hbar(\alpha) \diff \alpha\,.
\]
Thus by convexity of the $W_p$ distance and defining $\nu := \mcal{U}_T\rho$ we have
\begin{align*}
d_{W_p} (\Husimi_{\mcal{U}_T\rho}, \Husimi_{\mcal{U}_T\circ\mcal{N}[\rho]})^p
&\leq
\int \gamma_\hbar(\alpha) d_{W_p}(\Husimi_{\mcal{U}_T\rho}, \Husimi_{\mcal{U}_T\circ \mcal{T}_\alpha \rho})^p \diff \alpha \\
&=
\int \gamma_\hbar(\alpha) d_{W_p}(\Husimi_{\nu}, \Husimi_{W_{T,\alpha} \nu W_{T,\alpha}^\dagger})^p \diff \alpha \\
&\leq C e^{C_p T}
\int \gamma_\hbar(\alpha) ((\hbar^{1/2} + |\alpha|) (1 + |\alpha|^{4D}\hbar^{-2D}))^p \diff \alpha \\
&\lsim  C e^{C_p T} \hbar^{p/2}.
\end{align*}
In the second to last line we have used the bound~\eqref{eq:husimi-shift} from Theorem~\ref{thm:main-husimi}.

\subsection{Proof of Corollary~\ref{cor:LipEgo} from Theorem~\ref{thm:Wp-egorov}}
\label{sec:corpf}
Let $G$ satisfy $\partial_a G\in\mcal{S}_{\frac12}$ for all $a\in[2D]$.
It suffices to show that for any density matrix $\rho$,
\begin{equation}
\label{eq:G-duality}
|\trace[\rho (\mcal{U}_T(\Op_\hbar(G)) - \Op_\hbar(\Phi_TG))]|\leq e^{CT} \hbar^{1/2}.
\end{equation}

In the special case $p=1$, Theorem~\ref{thm:Wp-egorov} combined with Kantorovich-Rubinstein duality implies
\[
\Big|\int G (\Husimi_{\mcal{U}_T\rho} - \Phi_T \Husimi_\rho) \diff \alpha\Big| \leq e^{CT} \hbar^{1/2}.
\]
Using Wigner-Weyl duality and the fact that $\Husimi_\rho = \Wigner_\rho \ast \gamma_\hbar$, this inequality
implies
\[
|\trace[\rho (\mcal{U}_T(\Op_\hbar(G \ast \gamma_\hbar)) - \Op_\hbar((\Phi_T G) \ast \gamma_\hbar))]|
\leq e^{CT}\hbar^{1/2}.
\]
Thus to prove~\eqref{eq:G-duality} it suffices to show that
\[
\|\,\mcal{U}_T[\Op_\hbar (G\ast \gamma_\hbar -G)]\|_{\rm op}
+
\|\Op_\hbar ( \Phi_T G\ast \gamma_\hbar - \Phi_T G)\|_{\rm op} \leq e^{CT}\hbar^{1/2}.
\]
Observe that $\mcal{U}_T$ preserves the operator norm, so a bound on each of the terms on the left hand side
will then follow from the following estimate (which is a slightly stronger form of the Calder\'on-Vaillancourt theorem):
\[
\|\Op_\hbar(F)\|_{\rm op} \lsim \|F\|_{C^0} + \hbar^{(2D+1)/2}\|\partial^{2D+1} F\|_{C^0}\,.
\]
From the Lipschitz bound on $G$ we have
\[
\|G - G\ast \gamma_\hbar\|_{C^0} \lsim \hbar^{1/2} \|\partial G\|_{C^0}\,,
\]
and in general for higher order derivatives we have
\[
\|\partial^k G - \partial^k G\ast \gamma_\hbar\|_{C^0} \lsim \hbar^{1/2} \|\partial^{k+1} G\|_{C^0}
\lsim \hbar^{1/2} \hbar^{-k/2}.
\]
Thus, for example, it follows that
\[
\|\,\mcal{U}_T[\Op_\hbar (G\ast \gamma_\hbar -G)]\|_{\rm op}
=
\|\Op_\hbar (G\ast \gamma_\hbar -G)\|_{\rm op} \lsim \hbar^{1/2}.
\]
A similar estimate applies to $\Phi_T G - \Phi_T G\ast \gamma_\hbar$, up to a loss of $e^{(2D+1)C T}$.

\section*{Acknowledgements}

The authors would like to thank Zhenhao Li and Toan Nguyen for valuable discussions.  JC is supported by an Alfred P. Sloan Foundation Fellowship.  FH is supported by the National Science Foundation under Award No.~2303094.

\printbibliography

\end{document}